\documentclass[submission]{eptcs}
\providecommand{\event}{FICS 2012} 

\usepackage{amsmath}
\usepackage{amsthm,amssymb,latexsym}
\usepackage{bussproofs}

\newcommand{\proofrule}[3]
{
\displaystyle{\frac{#1}{#2}} \; #3
}

\theoremstyle{plain}
\newtheorem{theorem}{Theorem}
\newtheorem{lemma}[theorem]{Lemma}
\newtheorem{corollary}[theorem]{Corollary}

\theoremstyle{definition}
\newtheorem{definition}[theorem]{Definition}

\def\atoms{\textsc{Prop}}

\def\langnull{\mathcal{L}}
\def\langaux{\mathcal{L}_{\Omega}}

\def\clorule{\textsf{clo}}
\def\cutrule{\textsf{cut}}

\def\OmegaT{\tilde{\Omega}}
\def\nurule{\omega}

\def\indrule{\textsf{ind}}

\newcommand{\mathth}[1]{\mathbf{#1}}

\def\sys{\mathth{M}}
\def\sysi{\mathth{M}^{\omega}}
\def\syso{\mathth{M}^{\omega,\Omega}}

\def\nuaux{\nu'}

\def\prov#1#2{\mathrel{\vrule width 0.3pt height 8pt depth
2pt\mkern-2mu
\textstyle{\frac{\hskip0.5ex\raise2pt\hbox{$\scriptstyle#1$}\hskip0.5ex}
{\hskip0.5ex#2\hskip0.5ex}}}}

\newcommand{\onot}[1]{\overline{#1}}
\def\lev{\mathsf{lev}}

\title{Cut-elimination for the mu-calculus with one variable}
\author{Grigori Mints
\institute{Dept.\ of Philosophy\\Stanford University\\USA}
\email{gmints@stanford.edu}
\and
Thomas Studer
\institute{Inst.\ of Computer Science and Appl.\ Math.\\
University of Bern\\Switzerland}
\email{tstuder@iam.unibe.ch}
}
\def\titlerunning{Cut-elimination for the mu-calculus with one variable}
\def\authorrunning{G.~Mints \& T.~Studer}

\begin{document}

\maketitle

\begin{abstract}
We establish syntactic cut-elimination for the one-variable fragment of the modal mu-calculus. Our method is based on a recent cut-elimination technique by Mints that makes use of Buchholz' $\Omega$-rule.
\end{abstract}

\section{Introduction}

The propositional modal $\mu$-calculus is a well-established modal fixed point logic that includes fixed points for arbitrary positive formulae. Thus it subsumes many temporal logics (with an always operator), epistemic logics (with a common knowledge operator), and program logics (with an iteration operator). 

Making use of the finite model property, Kozen~\cite{koz88} introduces a sound and complete infinitary system for the modal $\mu$-calculus.
In this system greatest fixed points are introduced by means of the $\omega$-rule that has a premise for each finite approximation of the greatest fixed point. 
J\"ager et al.~\cite{jks:mu} show by {\em semantic} methods that the cut rule is admissible in this kind of infinitary systems. So far, however, there is no {\em syntactic} cut-elimination procedure available for the modal $\mu$-calculus. It is our aim in this paper to present an effective cut-elimination method for the one-variable fragment of the $\mu$-calculus.

There are already a few results available on syntactic cut-elimination for modal fixed point logics. Most of them make use of deep inference where rules may not only be applied to outermost connectives but also deeply inside formulae. The first result of this kind has been obtained by Pliuskevicius~\cite{Pliuskevicius91} who presents a syntactic cut-elimination procedure for linear time temporal logic.
Br\"unnler and Studer~\cite{bs09b} employ nested sequents to develop a cut-elimination procedure for the logic of common knowledge. Hill and Poggiolesi~\cite{hipo10} use a similar approach to establish effective cut-elimination for propositional dynamic logic. A generalization of this method is studied in~\cite{bs11} where it is also shown that it cannot be extended to fixed points that have a $\Box$-operator in the scope of a $\mu$-operator. Fixed points of this kind occur, for instance, in $\mathsf{CTL}$ in the form of universal path quantifiers.

Thus we need a more general approach to obtain syntactic cut-elimination for the modal $\mu$-calculus. A standard proof-theoretic technique to deal with inductive definitions and fixed points is Buchholz' $\Omega$-rule~\cite{buchholz81,buschue88}. J\"ager and Studer~\cite{js11} present a formulation of the $\Omega$-rule for non-iterated modal fixed point logic and they obtain cut-elimination for positive formulae of this logic.
In order to overcome this restriction to positive formulae, Mints~\cite{mintsMu} introduces an $\Omega$-rule that has a wider set of premises, which enables him to obtain full cut-elimination for non-iterated modal fixed point logic. 

Mints' cut-elimination algorithm makes use of, in addition to ideas from~\cite{buchholz01}, a new tool presented in~\cite{mintsMu}.
It is based on the distinction, see \cite{takeuti87}, between implicit and explicit occurrences of formulae in a derivation with cut. If an occurrence of a formula is traceable to the endsequent of the derivation, then it is called explicit. If it is traceable to a cut-formula, then it is an implicit occurrence.  

Implicit and explicit occurrences of greatest fixed points are treated differently in the translation of the induction rule to the infinitary system.
An instance of the induction rule that derives a sequent $\nu X.A, B$ goes to an instance of the $\omega$-rule if $\nu X.A$ is explicit. 
Otherwise, if $\nu X.A$ is traceable to a cut-formula, the induction rule is translated to an instance of the $\Omega$-rule that is preserved until the last stage of cut-elimination. At that stage, called collapsing, the $\Omega$-rule is eliminated completely.
 
In the present paper we show that this method can be extended to
a $\mu$-calculus with iterated fixed points. Hence we obtain complete syntactic cut-elimination for the one-variable fragment of the modal $\mu$-calculus. Our infinitary system is completely cut-free in the sense that there are not only no cut rules in the system but also no embedded cuts. Thus our cut-free system enjoys the subformula property. This is in contrast to the recent cut-elimination results by Baelde~\cite{DBLP:journals/corr/abs-0910-3383} and by Tiu and Momigliano~\cite{DBLP:journals/corr/abs-1009-6171} for the finitary systems $\mu\mathsf{MALL}$ and $\mathsf{Linc}^-$, respectively, where the $\nu$-introduction rule and the co-induction rule contain embedded cuts, which results in the loss of the subformula property.

%

\section{Syntax and semantics}

We first introduce the language $\langnull$.
We start with a countable set $\atoms$ of atomic propositions $p_i$ and their negations $\onot{p_i}$.
We use $P$ to denote an arbitrary element of $\atoms$.
Moreover, we will use a special variable $X$.

\begin{definition} 
{\em Operator forms} $A, B, \ldots$ are given by the following grammar:
\[
A :== p_i \ |\  \onot{p_i} \ |\ X \ |\ A \land A \ |\ A \lor A \ |\ \Box A \ |\ \Diamond A \ |\ \mu X.A \ |\ \nu X.A .
\] 
{\em Formulae} $F$ are defined by:
\[
F :== p_i \ |\  \onot{p_i}  \ |\ F \land F \ |\ F \lor F \ |\ \Box F \ |\ \Diamond F \ |\ \mu X.A \ |\ \nu X.A .
\] 
\end{definition}

The fixed point operators $\mu$ and $\nu$ bind the variable $X$ and, therefore, we will talk 
of free and bound occurrences of $X$. Hence a formula is an operator form without free occurrences of $X$.


The negation of an operator form is inductively defined as follows.
\begin{enumerate}
\item $\lnot p_i := \onot{p_i}$ and $\lnot \onot{p_i} := p_i$
\item $\lnot X := X$
\item $\lnot (A \land B) := \lnot A \lor \lnot B$ and $\lnot (A \lor B) := \lnot A \land \lnot B$
\item $\lnot \Box A := \Diamond \lnot A$ and $ \lnot \Diamond A := \Box \lnot A$
\item $\lnot \mu X.A := \nu X.\lnot A$ and $\lnot \nu X.A := \mu X.\lnot A$
\end{enumerate}

Note that negation is well-defined: the negation of an $X$-positive operator form is again $X$-positive since we have $\lnot X := X$. 
Thus, for example, 
\[
\lnot \mu X.\Box (p_i \land X) := \nu X.\lnot \Box (p_i \land X) := \nu X. \Diamond \lnot (p_i \land X) := \nu X. \Diamond ( \lnot p_i \lor \lnot X):= \nu X. \Diamond ( \onot{p_i} \lor  X).
\] 

For an arbitrary but fixed atomic proposition $p_i$ we set
%
$\top := p_i \lor \onot{p_i}$. 
If $A$ is an operator form, then we write $A(B)$ for the result of simultaneously substituting $B$ for 
every free occurrence of $X$ in $A$. We will also use finite iterations of operator forms, given as follows
\[
A^0(B) := B \text{ and } A^{k+1}(B):= A (A^k(B)).
\]

\section{System $\sys$}

System $\sys$  derives sequents, that are finite sets of formulae. We denote sequents by $\Gamma, \Sigma$ and use the following notation:
if $\Gamma := \{A_1, \ldots, A_n  \}$, then $\Diamond \Gamma := \{\Diamond A_1, \ldots,\Diamond A_n  \}$,
System $\sys$ consists of the axioms and rules given in Figure~\ref{fig:sys}.

\begin{figure}[h]
  \centering
\fbox{\parbox{.9\linewidth}{
\[
{\Gamma,P,\lnot P} 
\qquad\qquad 
{\Gamma,\mu X.A, \lnot \mu X.A} 
\]

\[
\proofrule{\Gamma,A,B}{\Gamma,A \lor B}{(\lor)}
\qquad\qquad
\proofrule{\Gamma,A \quad\quad\quad \Gamma,B}{\Gamma,A \land B}{(\land)}
\qquad\qquad
\proofrule{\Gamma, A}{\Diamond \Gamma, \Box A, \Sigma}{(\Box)}
\]

\[
\proofrule{\Gamma, A(\mu X.A)}{\Gamma, \mu X.A}{(\clorule)}
\qquad\qquad
\proofrule{\lnot A(B), B}{\lnot \mu X.A, B}{(\indrule)}
\qquad\qquad
\proofrule{\Gamma,A \quad\quad\quad \Gamma,\lnot A}{\Gamma}{(\cutrule)}
\]
}}
  \caption{System $\sys$}
\label{fig:sys}
\end{figure}



\section{System $\sysi$}

System $\sysi$ is an infinitary cut-free system for the modal $\mu$-calculus with one variable. It consists of the axioms and rules
 given in Figure~\ref{fig:sysi}.

\begin{figure}[h]
  \centering
\fbox{\parbox{.9\linewidth}{
\[
{\Gamma,P,\lnot P} 
\]

\[
\proofrule{\Gamma,A,B}{\Gamma,A \lor B}{(\lor)}
\qquad\qquad
\proofrule{\Gamma,A \quad\quad\quad \Gamma,B}{\Gamma,A \land B}{(\land)}
\qquad\qquad
\proofrule{\Gamma, A}{\Diamond \Gamma, \Box A, \Sigma}{(\Box)}
\]

\[
\proofrule{\Gamma, A(\mu X.A)}{\Gamma, \mu X.A}{(\clorule)}
\qquad\qquad
\proofrule{\Gamma, A^i(\top) \text{ for all natural numbers $i$}}{\Gamma,\nu X.A}{(\nurule)}
\]
}}
  \caption{System $\sysi$}
\label{fig:sysi}
\end{figure}

%
%

\section{System $\syso_k$}

In order to embed $\sys$ into $\sysi$, we need a family of intermediate systems $\syso_k$ that include additional rules to derive greatest fixed points that later will be cut away.

The language $\langaux$ extends  $\langnull$ by a 
new connective $\nuaux$ to denote those greatest fixed points.
Formally, $\langaux$ is given as follows.
Operator forms of $\langaux$ are defined like operator forms of $\langnull$
with the additional case
\begin{enumerate}
\item If $A$ is an operator form, then $\nuaux X.A$ is also an operator form. 
\end{enumerate}
A formula of $\langaux$ is an $\langaux$ operator form without free occurrence of $X$.
A formula is a {\em greatest fixed point} if it has the form $\nu X.A$ or $\nuaux X.A$.

\begin{definition}\label{def:level:1}
The level $\lev(A)$ of an operator form $A$ is the maximal nesting of fixed point operators in $A$. Formally we set:
\begin{enumerate}
\item $\lev(P) := \lev(X) :=  0$  for all $P$ in $\atoms$
\item $\lev(A \land B) := \lev(A \lor B) := \max(\lev(A), \lev(B))$
\item $\lev(\Box A ) := \lev(\Diamond A) := \lev(A)$
\item $\lev(\mu X.A ) := \lev(\nu X.A) := \lev(\nuaux X.A) := \lev(A) +1$ 
\end{enumerate}  
\end{definition}

The level of a sequent is the maximum of the levels of its formulae.
We say a formula (sequent) is {\em $k$-positive} if 
for all $\nuaux X.A$ occurring in it we have $\lev(\nuaux X.A)<k$.



When working in $\syso_k$, we will use the following notation:
the formula $A'$ is obtained from $A$ by replacing all occurrences of $\nu X$ in $A$ with $\nuaux X$.  

Let $k \geq 0$.
System $\syso_k$ consists of the axioms and rules of $\sysi$ (formulated in $\langaux$) and the additional rules: $\cutrule$, $\Omega_h$, and $\OmegaT_h$.
The $\cutrule$ rule is given as follows
\[
\proofrule{\Gamma,A' \quad\quad\quad \Gamma,(\lnot A)'}{\Gamma}{(\cutrule)},
\]
where $A$ is a formula with $\lev(A) \leq k$.
The rules $\Omega_h$ and $\OmegaT_h$ , where $1 \leq h \leq k$, are informally described as follows:
\begin{prooftree}
\def\extraVskip{4pt}
		\AxiomC{$\cdots$}
		\AxiomC{$\syso_{k-1} \prov{}{0}  \Delta,\, (\mu X.A)'$}
		\dashedLine
		\UnaryInfC{$\Delta,\, \Gamma$}
		\AxiomC{$\cdots$}
		\RightLabel{\quad$\Omega_h$}
		\TrinaryInfC{$\Gamma,\, (\lnot \mu X.A)'$}
\end{prooftree}
and
\begin{prooftree}
\def\extraVskip{4pt}
		\AxiomC{$\Gamma,\, (\mu X.A)' \qquad\qquad \cdots$}
		\AxiomC{$\syso_{k-1} \prov{}{0}  \Delta,\, (\mu X.A)'$}
		\dashedLine
		\UnaryInfC{$\Delta,\, \Gamma$}
		\AxiomC{$\cdots$}
		\RightLabel{\quad$\OmegaT_h$}
		\TrinaryInfC{$\Gamma$}
\end{prooftree}
where 
$\lev((\lnot \mu X.A)') = h$ and
$\Delta$ ranges over $h$-positive sequents such that there is a cut-free proof of the sequent $\Delta,\, (\mu X.A)'$ in $\syso_{k-1}$.

\begin{definition}
We use $\syso_k \prov{}{0} \Gamma$ to express that there is a cut-free derivation of $\Gamma$ in $\syso_k$.
\end{definition}
In a more formal notation we can state the $\Omega_h$-rule as follows. If for every $h$-positive sequent $\Delta$
\[
\syso_{k-1} \prov{}{0}  \Delta,\, (\mu X.A)' \quad \Longrightarrow \quad 
\syso_{k} \prov{}{}  \Delta,\, \Gamma,
\]
then
\[
\syso_{k} \prov{}{}  \Gamma,\, (\lnot \mu X.A)',
\]
and similarly for $\OmegaT_h$.

Note that System $\syso_0$ does not include $\Omega_h$- or  $\OmegaT_h$-rules. Hence we immediately get the following lemma.
 
\begin{lemma}\label{eq:mnull:1}
Let\/ $\Gamma$ be an $\langnull$ sequent. We have
\[
\syso_0 \prov{}{0} \Gamma \quad \Longrightarrow \quad \sysi \prov{}{} \Gamma.
\]
\end{lemma}

\section{Embedding}

In this section we present a translation from $\sys$-proofs into $\syso_k$-proofs. 
First we establish an auxiliary lemma.

%
%
%
%
\begin{lemma}\label{l:key:1}
For all natural numbers $h \leq k$ we have the following.
\begin{enumerate}
\item If\/ $\lev(\mu X.A)=h$, then\/ $\syso_k \prov{}{0} \mu X.A,\, \lnot \mu X.A$.
\item If\/ $\lev(A)=h$, then\/ $\syso_k \prov{}{0} \Gamma,\, A' \quad\Longrightarrow\quad \syso_k \prov{}{0} \Gamma,\, A$.
\item If $\lev(\mu X.A)=h$, then $\syso_k \prov{}{0} \mu X.A,\, (\lnot \mu X.A)'$.
\item If\/ $\lev(A)=h$, then\/ $\syso_k \prov{}{0} B,\,C \quad\Longrightarrow\quad \syso_k \prov{}{0} (\lnot A)(B),\, A(C)$.
\item If\/ $\lev(A)=h$, then\/ $\syso_k \prov{}{0} B,\,C' \quad\Longrightarrow\quad \syso_k \prov{}{0} (\lnot A)(B),\, A'(C')$.
\end{enumerate}
\end{lemma}
\begin{proof}
The five statements are shown simultaneously by induction on $h$.
For space considerations we show only one particular case of the second statement, which is shown by induction on the derivation of
$ \Gamma,\, A'$ and a case distinction on the last rule.
Assume the last rule is an instance of $\Omega_h$ with main formula $A'$.
	We have $A' = (\nu X.A_0)'$ with $\lev(A_0)<h$.
	By the premise of the $\Omega_h$-rule we have for all $h$-positive sequents $\Delta$
	\begin{equation}\label{eq:simind:2}
	\syso_{k-1} \prov{}{0} \Delta, (\mu X.\lnot A_0)' \quad \Longrightarrow \quad \syso_k \prov{}{0} \Delta, \Gamma.
	\end{equation}
	Trivially we have
	\begin{equation}\label{eq:simind:4}
	\syso_k \prov{}{0} \top, \Gamma.
	\end{equation}
	We also have
	\begin{equation*}\label{eq:simind:3}
	\syso_{k-1} \prov{}{0} \top, (\mu X.\lnot A_0)'
	\end{equation*}
	from which we get by the induction hypothesis for the fifth claim of this lemma
%
%
	\[
	\syso_{k-1} \prov{}{0} A_0(\top), (\lnot A_0)'((\mu X.\lnot A_0)').
	\]
	An application of $\clorule$ yields
	\begin{equation*}\label{eq:simind:5}
	\syso_{k-1} \prov{}{0} A_0(\top), (\mu X.\lnot A_0)'.
	\end{equation*}
	By \eqref{eq:simind:2} we get
	\begin{equation}\label{eq:simind:6}
	\syso_k \prov{}{0} A_0(\top), \Gamma.
	\end{equation}
	Note that \eqref{eq:simind:4} and \eqref{eq:simind:6} are the first two premises of an instance of $\nurule$. By further iterating this we obtain for all $i$
	\[
	\syso_k \prov{}{0} A_0^i(\top), \Gamma.
	\]
	Hence an application of $\nurule$ yields 
	\[
	\syso_k \prov{}{0} \nu X.A_0, \Gamma.
	\]
\end{proof}

We will need a  certain form of the induction rule in $\syso_k$, which we are going to derive next. We write $\Sigma[(\mu X.A)':=B]$ for the result of simultaneously replacing in every formula in $\Sigma$ every occurrence of $(\mu X.A)'$ with $B$.


\begin{lemma}\label{l:clo:1}
Let $A$ be an operator form with $\lev(\nu X.A) \leq k$.
Let $\Delta, \Sigma_1, \Sigma_2$ be $h$-positive sequents 
and let $B$ be a formula with $\lev(B) \leq k$. 
Assume that
\begin{equation*}\label{eq:assInd:1}
\syso_k \prov{}{} (\lnot A(B))' ,\, B  \quad \text{and} \quad 
\syso_k \prov{}{} (\lnot A(B))' ,\, B'. 
\end{equation*}
Then we have, if
\[
\syso_{k-1} \prov{}{0} \Delta,\, \Sigma_1,\, \Sigma_2 
\]
then
\[
\syso_k \prov{}{} \Delta,\, \Sigma_1[(\mu X.A)':=B],\, \Sigma_2[(\mu X.A)':=B']. 
\]
\end{lemma}

\begin{lemma}\label{l:IndToOmega:1}
Let $A$ be an operator form with $\lev(\nu X.A) \leq k$. 
Further let $B$ be an arbitrary formula with $\lev(B) \leq k$. 
Assume that
\begin{equation*}\label{eq:assClo:1}
\syso_k \prov{}{} (\lnot A(B))' ,\, B  \quad \text{and} \quad 
\syso_k \prov{}{} (\lnot A(B))' ,\, B'. 
\end{equation*}
Then we have
\[ 
\syso_k \prov{}{} (\lnot \mu X.A)'  ,\, B   \quad \text{and} \quad 
\syso_k \prov{}{} (\lnot \mu X.A)'  ,\, B'. 
\]
\end{lemma}
\begin{proof}
Let $h = \lev(\nu X.A)$.
In view of our assumptions and the previous lemma we know that for all $h$-positive sequents $\Delta$
\[ 
\syso_{k-1} \prov{}{0} \Delta,\, (\mu X.A)' \quad\Longrightarrow\quad \syso_k \prov{}{} \Delta,\, B.
\]
Hence by an application of the $\Omega_h$-rule we conclude $\syso_k \prov{}{} (\lnot \mu X.A)' ,\, B$. 
Similarly, we can derive $\syso_k \prov{}{} (\lnot \mu X.A)' ,\, B'$. 
\end{proof}

\begin{theorem}\label{th:em:1}
Let\/ $\Gamma$ be a sequent of $\langnull$.
Assume $\sys \prov{}{} \Gamma$ and assume further for any sequent $\Delta$ occurring in that proof we have $\lev(\Delta) \leq k$.
Then we have $\syso_k \prov{}{} \Gamma$. 
\end{theorem}
\begin{proof}
An operation $\sigma$ on sequents is called '-operation if $\sigma(\Gamma,A_1,\ldots,A_n) = \Gamma,A'_1,\ldots,A'_n$.
%
The result of applying $\sigma$ to a sequent $\Gamma$ is denoted $\Gamma^\sigma$.

To establish the theorem, we show by induction on the depth of the $\sys$-proof that 
for all '-operations $\sigma$, we have 
$\syso_k \prov{}{} \Gamma^\sigma$. 
We distinguish the following cases for the last rule.
\begin{enumerate}
\item $\Gamma$ is an axiom different from $\Gamma_0, \mu X.A, \lnot \mu X.A$. Then $\Gamma^\sigma$ is an axiom of $\syso_k$, too.
\item $\Gamma$ is $\Gamma_0, \mu X.A, \lnot \mu X.A$.
	Then $\Gamma^\sigma$ follows either by the first or the third claim of Lemma \ref{l:key:1} depending on whether $\lnot \mu X.A$ is replaced by $\sigma$ or not.
\item The last rule is an instance of $\land$, $\lor$, $\Box$ or $\clorule$.
	We can apply the same rule in $\syso_k$.
\item The last rule is a cut
	\[
	\proofrule{\Gamma,A \quad\quad\quad \Gamma,\lnot A}{\Gamma}{}. 
	\]
      We extend the current '-operation $\sigma$ to a '-operation $\tau$ such that
      $(\Gamma, A)^\tau = \Gamma^\sigma, A'$ and 
	$(\Gamma, \lnot A)^\tau = \Gamma^\sigma, (\lnot A)'$ 
	By the induction hypothesis for the '-operation  $\tau$ we obtain 
	$\syso_k \prov{}{} \Gamma^\sigma, A'$ as well as 
	$\syso_k \prov{}{} \Gamma^\sigma, (\lnot A)'$.
	With an instance of cut we get
	$\syso_k \prov{}{} \Gamma^\sigma$.
\item The last rule is an instance of the induction rule. Then the endsequent has the form 
$\lnot \mu X.A,\, B$ which is $\nu X. \lnot A,\, B$. 
There are two possible cases.
	\begin{enumerate}
	\item The principal occurrence of $\nu X. \lnot A$ is not changed by $\sigma$.
		By the induction hypothesis we can derive 
		$(\lnot A(B))',\, B^\sigma$ and 
		$(\lnot A(B))',\, B'$.
		We obtain our claim by the following proof.
		\begin{prooftree}
		\def\extraVskip{4pt}
		\AxiomC{$\cdots$}

		\AxiomC{}
		\RightLabel{I.H.}		
		\UnaryInfC{$(\lnot A(B))',\, B^\sigma$}

		\AxiomC{}
		\RightLabel{I.H.}		
		\UnaryInfC{$(\lnot A(B))',\, B'$}
		\AxiomC{$\top, B'$}
		\RightLabel{L.~\ref{l:key:1}}		
		\UnaryInfC{$(\lnot A)(\top), (A(B))'$}
		\RightLabel{\quad$\cutrule$}		
		\BinaryInfC{$(\lnot A)(\top), B'$}
		\noLine
		\UnaryInfC{$\vdots$}
		\noLine
		\UnaryInfC{$(\lnot A)^i(\top), B'$}
		\RightLabel{L.~\ref{l:key:1}}		
		\UnaryInfC{$(\lnot A)^{i+1}(\top), (A(B))'$}
		\RightLabel{$\cutrule$}		
		\BinaryInfC{$(\lnot A)^{i+1}(\top), B^\sigma$}

		\AxiomC{$\cdots$}
		\RightLabel{$\nurule$}
		\TrinaryInfC{$\nu X. \lnot A,\, B^\sigma$}
\end{prooftree}
%
%

	%
	\item The principal occurrence of $\nu X. \lnot A$
	is changed by $\sigma$. 
      Let $\tau_1, \tau_2$ be '-operations such that
	\[
	(\lnot A(B),  B)^{\tau_1} = (\lnot A(B))', B
	\]
	and
	\[
	(\lnot A(B), B)^{\tau_2} = (\lnot A(B))', B'.
	\]
	By the induction hypothesis for $\tau_1$ and $\tau_2$ we obtain 
      \[
      \syso_k \prov{}{} (\lnot A(B))' ,\, B \quad \text{and} \quad
      \syso_k \prov{}{} (\lnot A(B))' ,\, B'.
	\]
	We apply 
	Lemma~\ref{l:IndToOmega:1} and conclude
	$\syso_k \prov{}{}  (\lnot \mu X.A)' ,\, B^\sigma$.
	\qedhere
	\end{enumerate}
\end{enumerate}
\end{proof}

\section{Cut elimination}

We eliminate instances of $\cutrule$ in the standard way, see for instance~\cite{buchholz01,mintsMu}, by pushing them up the derivation. When an instance of $\cutrule$ with cut formulae $(\mu X.A)'$ and $(\lnot \mu X.A)'$ meets the instance of $\Omega_h$ that introduces $(\lnot \mu X.A)'$, this pair of inferences is replaced by 
$\OmegaT_h$. 

\begin{lemma}[Cut-elimination]\label{l:el:1}
If\/ $\syso_k \prov{}{} \Gamma$, then\/ $\syso_k \prov{}{0} \Gamma$.
\end{lemma}

The cut-elimination process terminates in a formally cut-free derivation that may contain instances of $\OmegaT_h$-rules.
Now we show that these instances of $\OmegaT_h$ also can be eliminated.


\begin{lemma}[Collapsing]\label{l:col1:n}
Let\/ $\Gamma$ be an $(h+1)$-positive sequent.
If\/ $\syso_k \prov{}{0} \Gamma$, then\/ $\syso_{h} \prov{}{0} \Gamma$.
\end{lemma}
\begin{proof}
By transfinite induction on the derivation in $\syso_k$.
The only interesting case is when the last rule is an instance of $\OmegaT_l$ for $h< l \leq k$ as follows
\begin{prooftree}
\def\extraVskip{4pt}
		\AxiomC{$\Gamma,\, (\mu X.A)' \qquad\qquad \cdots$}
		\AxiomC{$\syso_{l-1} \prov{}{0}  \Delta,\, (\mu X.A)'$}
		\dashedLine
		\UnaryInfC{$\Delta,\, \Gamma$}
		\AxiomC{$\cdots$}
		\RightLabel{\quad$\OmegaT_l$}
		\TrinaryInfC{$\Gamma$}
\end{prooftree}
Note that $\Gamma, (\mu X.A)'$ is $l$-positive.
Thus by the induction hypothesis we get
\begin{equation}\label{eq:col:1:n} 
\syso_{l-1} \prov{}{0} \Gamma,\, (\mu X.A)' .
\end{equation}
Moreover, also by the induction hypothesis we get for all $(h+1)$-positive $\Delta$
\begin{equation}\label{eq:col:2:n} 
\syso_{l-1} \prov{}{0} \Delta,\, (\mu X.A)' \quad \Longrightarrow \quad  
\syso_{h} \prov{}{0} \Delta,\,\Gamma.  
\end{equation}
Now we plug \eqref{eq:col:1:n} in \eqref{eq:col:2:n} and obtain $\syso_{h} \prov{}{0} \Gamma$ as required.
\end{proof}

We now have all ingredients ready for our main result.

\begin{corollary}
Let\/ $\Gamma$ be an  $\langnull$-sequent. We have
\[
\sys\prov{}{} \Gamma \quad \Longrightarrow \quad
\sysi \prov{}{} \Gamma.
\]
\end{corollary}
\begin{proof}
Assume $\sys\prov{}{} \Gamma$.
By Theorem~\ref{th:em:1} we get $\syso_k\prov{}{} \Gamma$ for some $k$.
By cut-elimination we obtain $\syso_k\prov{}{0} \Gamma$.
Then collapsing yields $\syso_0\prov{}{0} \Gamma$
which finally gives us $\sysi \prov{}{} \Gamma$ by Lemma~\ref{eq:mnull:1}. 
\end{proof}

\bibliographystyle{eptcs}
\bibliography{references}

\begin{thebibliography}{10}
\providecommand{\bibitemdeclare}[2]{}
\providecommand{\urlprefix}{Available at }
\providecommand{\url}[1]{\texttt{#1}}
\providecommand{\href}[2]{\texttt{#2}}
\providecommand{\urlalt}[2]{\href{#1}{#2}}
\providecommand{\doi}[1]{doi:\urlalt{http://dx.doi.org/#1}{#1}}
\providecommand{\bibinfo}[2]{#2}

\bibitemdeclare{article}{DBLP:journals/corr/abs-0910-3383}
\bibitem{DBLP:journals/corr/abs-0910-3383}
\bibinfo{author}{David Baelde} (\bibinfo{year}{2009}):
  \emph{\bibinfo{title}{Least and greatest fixed points in linear logic}}.
\newblock {\sl \bibinfo{journal}{CoRR}} \bibinfo{volume}{abs/0910.3383v4}.
\newblock \urlprefix\url{http://arxiv.org/abs/0910.3383v4}.

\bibitemdeclare{article}{bs09b}
\bibitem{bs09b}
\bibinfo{author}{Kai Br{\"u}nnler} \& \bibinfo{author}{Thomas Studer}
  (\bibinfo{year}{2009}): \emph{\bibinfo{title}{Syntactic cut-elimination for
  common knowledge}}.
\newblock {\sl \bibinfo{journal}{Annals of Pure and Applied Logic}}
  \bibinfo{volume}{160}(\bibinfo{number}{1}), pp. \bibinfo{pages}{82--95},
  \doi{10.1016/j.apal.2009.01.014}.

\bibitemdeclare{misc}{bs11}
\bibitem{bs11}
\bibinfo{author}{Kai Br{\"u}nnler} \& \bibinfo{author}{Thomas Studer}
  (\bibinfo{year}{preprint}): \emph{\bibinfo{title}{Syntactic cut-elimination
  for a fragment of the modal mu-calculus}}.

\bibitemdeclare{incollection}{buchholz81}
\bibitem{buchholz81}
\bibinfo{author}{Wilfried Buchholz} (\bibinfo{year}{1981}):
  \emph{\bibinfo{title}{The {$\Omega_{\mu +1}$}-rule}}.
\newblock In \bibinfo{editor}{Wilfried Buchholz}, \bibinfo{editor}{Solomon
  Feferman}, \bibinfo{editor}{Wolfram Pohlers} \& \bibinfo{editor}{Wilfried
  Sieg}, editors: {\sl \bibinfo{booktitle}{Iterated Inductive Definitions and
  Subsystems of Analysis: Recent Proof Theoretic Studies}}, {\sl
  \bibinfo{series}{Lecture Notes in Mathematics}} \bibinfo{volume}{897},
  \bibinfo{publisher}{Springer}, pp. \bibinfo{pages}{189--233},
  \doi{10.1007/BFb0091898}.

\bibitemdeclare{article}{buchholz01}
\bibitem{buchholz01}
\bibinfo{author}{Wilfried Buchholz} (\bibinfo{year}{2001}):
  \emph{\bibinfo{title}{Explaining the {G}entzen-{T}akeuti reduction steps: a
  second-order system}}.
\newblock {\sl \bibinfo{journal}{Archive for Mathematical Logic}}
  \bibinfo{volume}{40}(\bibinfo{number}{4}), pp. \bibinfo{pages}{255--272},
  \doi{10.1007/s001530000064}.

\bibitemdeclare{book}{buschue88}
\bibitem{buschue88}
\bibinfo{author}{Wilfried Buchholz} \& \bibinfo{author}{Kurt Sch\"utte}
  (\bibinfo{year}{1988}): \emph{\bibinfo{title}{Proof Theory of Impredicative
  Subsystems of Analysis}}.
\newblock \bibinfo{publisher}{Bibliopolis}.

\bibitemdeclare{article}{hipo10}
\bibitem{hipo10}
\bibinfo{author}{Brian Hill} \& \bibinfo{author}{Francesca Poggiolesi}
  (\bibinfo{year}{2010}): \emph{\bibinfo{title}{A Contraction-free and Cut-free
  Sequent Calculus for Propositional Dynamic Logic}}.
\newblock {\sl \bibinfo{journal}{Studia Logica}}
  \bibinfo{volume}{94}(\bibinfo{number}{1}), pp. \bibinfo{pages}{47--72},
  \doi{10.1007/s11225-010-9224-z}.

\bibitemdeclare{article}{jks:mu}
\bibitem{jks:mu}
\bibinfo{author}{Gerhard J{\"a}ger}, \bibinfo{author}{Mathis Kretz} \&
  \bibinfo{author}{Thomas Studer} (\bibinfo{year}{2008}):
  \emph{\bibinfo{title}{Canonical completeness for infinitary $\mu$}}.
\newblock {\sl \bibinfo{journal}{Journal of Logic and Algebraic Programming}}
  \bibinfo{volume}{76}(\bibinfo{number}{2}), pp. \bibinfo{pages}{270--292},
  \doi{10.1016/j.jlap.2008.02.005}.

\bibitemdeclare{article}{js11}
\bibitem{js11}
\bibinfo{author}{Gerhard J{\"a}ger} \& \bibinfo{author}{Thomas Studer}
  (\bibinfo{year}{2011}): \emph{\bibinfo{title}{A {B}uchholz rule for modal
  fixed point logics}}.
\newblock {\sl \bibinfo{journal}{Logica Universalis}} \bibinfo{volume}{5}, pp.
  \bibinfo{pages}{1--19}, \doi{10.1007/s11787-010-0022-1}.

\bibitemdeclare{article}{koz88}
\bibitem{koz88}
\bibinfo{author}{Dexter Kozen} (\bibinfo{year}{1988}): \emph{\bibinfo{title}{A
  finite model theorem for the propositional $\mu$--calculus}}.
\newblock {\sl \bibinfo{journal}{Studia Logica}}
  \bibinfo{volume}{47}(\bibinfo{number}{3}), pp. \bibinfo{pages}{233--241},
  \doi{10.1007/BF00370554}.

\bibitemdeclare{article}{mintsMu}
\bibitem{mintsMu}
\bibinfo{author}{Grigori Mints} (\bibinfo{year}{to appear}):
  \emph{\bibinfo{title}{Effective Cut-elimination for a fragment of Modal
  mu-calculus}}.
\newblock {\sl \bibinfo{journal}{Studia Logica}} .

\bibitemdeclare{inproceedings}{Pliuskevicius91}
\bibitem{Pliuskevicius91}
\bibinfo{author}{Regimantas Pliuskevicius} (\bibinfo{year}{1991}):
  \emph{\bibinfo{title}{Investigation of Finitary Calculus for a Discrete
  Linear Time Logic by means of Infinitary Calculus}}.
\newblock In: {\sl \bibinfo{booktitle}{Baltic Computer Science, Selected
  Papers}}, \bibinfo{publisher}{Springer}, pp. \bibinfo{pages}{504--528},
  \doi{10.1007/BFb0019366}.

\bibitemdeclare{book}{takeuti87}
\bibitem{takeuti87}
\bibinfo{author}{Gaisi Takeuti} (\bibinfo{year}{1987}):
  \emph{\bibinfo{title}{Proof Theory}}.
\newblock \bibinfo{publisher}{North-Holland}.

\bibitemdeclare{article}{DBLP:journals/corr/abs-1009-6171}
\bibitem{DBLP:journals/corr/abs-1009-6171}
\bibinfo{author}{Alwen Tiu} \& \bibinfo{author}{Alberto Momigliano}
  (\bibinfo{year}{2010}): \emph{\bibinfo{title}{Cut Elimination for a Logic
  with Induction and Co-induction}}.
\newblock {\sl \bibinfo{journal}{CoRR}} \bibinfo{volume}{abs/1009.6171v1}.
\newblock \urlprefix\url{http://arxiv.org/abs/1009.6171v1}.

\end{thebibliography}

\end{document}